\documentclass[a4paper,thm-restate,cleveref]{lipics-v2021}

\usepackage{latexsym}
\usepackage{amsmath}
\usepackage{amssymb}
\usepackage{amsthm}
\usepackage{thmtools}
\usepackage{epsfig}
\usepackage{ulem} 

\newcommand\numberthis{\addtocounter{equation}{1}\tag{\theequation}}

\usepackage{enumerate}
\usepackage{listings}
\usepackage{xfrac}

\usepackage[color, leftbars]{changebar}

\usepackage[ruled,vlined]{algorithm2e}
\usepackage[linewidth=1.2pt,linecolor=black,nobreak=true]{mdframed}
\mdfsetup{frametitlealignment=\center}
\newcommand{\remove}[1]{}

\usepackage{preamble}

\newcommand{\red}[1]{#1}

\usepackage{cleveref}

\title{Low-Crossing Spanning Paths via Quasi-Uniform Sampling and their Applications}
\date{}



\author{M\'onika Csik\'os}{ Universit\'e Paris Cit\'e, IRIF, CNRS UMR 8243 }{csikos@irif.fr}{}{}

\author{Nabil H. Mustafa}{ Universit\'e Sorbonne Paris Nord, Laboratoire LIPN, CNRS 7030}{nabil.mustafa@univ-paris13.fr}{}{}

\authorrunning{M. Csik\'os and N. H. Mustafa} 
\acknowledgements{}
\Copyright{M. Csik\'os and N. H. Mustafa}

\begin{document}

\maketitle

\begin{abstract}


\noindent Given a set system $(X, \mathcal S)$, we consider the problems of constructing perfect matchings and spanning paths 
(also known as Welzl orders) of $X$ with low crossing numbers with respect to $\mathcal S$. 
Matchings with low crossing numbers were originally introduced in the late 1980s in the seminal works of Welzl~\cite{Wel88, Wel92} and Chazelle-Welzl~\cite{CW89}, who showed that any set system $(X, \mathcal S)$ with bounded dual VC-dimension $d$ has a matching with crossing number $O\left(|X|^{1-1/d}\right)$, which   is optimal. 

{\color{myred}
Low-crossing matchings and spanning paths have found several applications outside of their initial realm of computational geometry, including algorithmic graph theory, model checking, and game theory. In many of these applications, $|\mathcal S|$ is near-linear in $|X|$, which motivated the search for an improved construction for this regime. However, the proposed new algorithms provided only average-case bounds or sub-optimal crossing number guarantees for $d \geq 2$~\cite{CCGKLZ25,DK26}. }

We propose a new quasi-uniform sampling based framework. This allows us to obtain a simple algorithm with expected running time of $\tilde{O}\left(|X|^2 + |\mathcal S|\cdot |X|\right)$, which improves upon the previous-best running time of $\tilde{O} \left(|X|^{2+2/d} + |\mathcal S|\cdot|X|^{2/d} \right)$~\cite{CsM21} when $|\mathcal S| = o\left(|X|^{1+2/d}\right)$.

Our result has several immediate consequences, including new faster algorithms for  constructing low-crossing spanning paths for set systems induced by balls, as well as for  neighborhood set systems of graphs.
Moreover, we get the first quadratic-time algorithm for combinatorial discrepancy, with nearly  optimal discrepancy (up to a $\log$ factor), for bounded dual VC-dimension set systems.
Beyond these applications, we believe that our idea could be applied to other algorithms that use the multiplicative weight update method.

\bigskip



\end{abstract}

\medskip
\noindent{\small
\begin{description}
    \item[AI usage] The authors declare that they did not use generative AI tools to obtain their results or to write this article.  An LLM was used to check the manuscript before submission, leading to minor technical and grammatical corrections.
\end{description}
}

\newpage 

\section{Introduction}
Let $(X, \S)$ be a set system (also known as hypergraph, concept class, or range space), where $X$ is a finite set of $n$ elements and $\S$ is a collection
of $m$ subsets of $X$.
To distinguish the sets of $\S$ from arbitrary subsets of $X$, we will refer to the elements of $\S$ as \emph{ranges}. 

A perfect matching of a set $X$ is a partition of $X$ into $|X|/2$ disjoint pairs\footnote{If $|X|$ is odd, then we partition $X$ into $\floor{|X|/2}$ disjoint pairs plus a singleton set.}.
Given a set system $(X,\S)$, we say that
a set $S \in \S$ \emph{crosses} a pair $\{x,y\} \subseteq X$  iff $|S \cap \{x,y\}| =1$.
The \emph{crossing number} of a perfect matching $M$
with respect to $\S$ is defined to be the maximum number of edges of $M$ crossed by any $S\in \S$.

Matchings with low crossing numbers were originally introduced by Welzl~\cite{Wel88, Wel92} for the special case where $X$ is a set of points in $\RR^d$ and   $\S$ is induced on $X$ by half-spaces. Welzl's result was then generalized by Chazelle and Welzl~\cite{CW89} to a broader class of set systems using the notion of the \emph{dual shatter function}.
\begin{definition*}[dual shatter function]
For a set system $(X,\S)$, the dual shatter function $\pi^*_\S$ of $(X,\S)$ is defined as follows.
For any $k \leq |\S|$, $\pi^*_\S(k)$ denotes the maximum number of equivalence classes on $X$ defined by a $k$-element subfamily $\R \subseteq \S$, where $x,y \in X$ are equivalent\footnote{In other words, the number of equivalence classes for a given $\R \subseteq \S$,  is  the number of non-empty cells, with respect to $X$, in the Venn diagram of the sets of $\R$.} w.r.t. $\R$ iff $x$ belongs to the same sets of $\R$ as $y$.
  \end{definition*}
 \textbf{Remark.}
 The family of set systems with    dual-shatter function bounded polynomially in $k$ includes  set systems with bounded dual VC-dimension. In particular,
a set system with dual VC-dimension at most $d$ has dual shatter function upper bounded by $O \left(k^d\right)$~\cite{Sa72,Sh72}.
It also includes most of the commonly-studied  geometric set systems such as the primal and dual set systems induced by half-spaces, balls, intersections/unions of bounded complexity geometric objects and more generally, algebraic varieties (see~\cite{matousek2013lectures}).

The method of Chazelle and Welzl, together with Haussler's packing lemma bound~\cite{Haussler92spherepacking} 
led to the following fundamental theorem on the existence of matchings with low crossing numbers.
 
\begin{thmx}[\cite{CW89,Haussler92spherepacking}]\label{thm:original-matching-thm}
    Let $(X,\S)$ be a set system with $n = |X|$, $m = |\S|$, and dual shatter function $\pi^*_{\S}(k)  = O(k^d)$.
Then there exists
    a perfect matching on $X$ with crossing number $O \left( {n^{1-1/d}} + \ln m \right)$.
\end{thmx}

\Cref{thm:original-matching-thm} has had numerous applications for geometric set systems, including
range searching, low-discrepancy colorings,  near-optimal sized
epsilon-approximations, to name a few.  We refer
the reader to these books~\cite{ Chazelle:2000:DMR:507108, MatDiscBook, PA95} for proofs and applications  of \Cref{thm:original-matching-thm}.
Moreover, the generality of \Cref{thm:original-matching-thm}  makes possible applications in several other fields of computer science: communication complexity and learning theory~\cite{AMY16},  algorithmic graph theory~\cite{DHV20}, logic and verification~\cite{DEMMPT23}, or game theory~\cite{HWY25}.

Low-crossing matchings are a key algorithmic tool, and improving their algorithmic complexity has been an active area of research~\cite{harp09, Chan12,GKM14, CsM21, DHV20, CM24}. For geometric set systems, the current best algorithms rely on spatial partitioning tools such as cuttings~\cite{Chan12} or polynomial partitioning~\cite{AMS13}. However, these geometric techniques cannot give improved algorithms for \Cref{thm:original-matching-thm} in its abstract generality. For the general case of set systems with polynomially bounded dual shatter functions, the current-best algorithm
for computing a matching of crossing number $O \left(n^{1-1/d} \right)$ takes expected time 
$O\left( mn^{2/d} + n^{2+2/d} \right)$~\cite{CsM21}. 
One can get faster algorithms at the cost of a sub-optimal crossing number. In particular, Csik\'os and Mustafa~\cite{CM24} showed that for any $\alpha \in(0,1]$, it is possible to compute matchings with crossing number $O\round{n^{1-\alpha/d}}$ in time $\aO \round{n^{1+\alpha + \frac{2\alpha}{d}} + |\S|\cdot n^{\frac{2\alpha}{d}}}$. {\color{myred} Recently, Dreier and Kuske~\cite{DK26} proposed an algorithm in the more restricted setting of set systems with primal and dual shatter functions in $O(k^d)$. Their guaranteed crossing number is  $O\round{n^{1-1/d^2}\log^2(n)}$, which is optimal if $d=1$ and the construction takes time $O(N \log N)$, where $N = n + \sum\limits_{S \in \S} |S|$.}

Finally, we make several remarks about our setup and related problems.
\begin{description}
    \item[Matchings vs paths vs trees. ] The  equivalence of low-crossing spanning paths and trees is a classical result, see e.g. Chapter 15 in the book of Pach and Agarwal~\cite{PA95}. We also note that by iteratively constructing low-crossing matchings, one can obtain spanning paths with asymptotically the same crossing number guarantees (Appendix~\ref{sec:proof_of_corollary_ref_cor_spanpath}). 
    \item[Worst case vs amortized cost. ] A related problem is that of computing a tree (or a matching, path) that minimizes the \emph{sum} of the crossings of the ranges, over all edges of the tree (instead of minimizing the maximum number of edges crossed by any set). This sum minimization problem can be reduced to the minimum spanning tree problem under the so-called crossing metric, which can be approximated in near-linear time for set systems induced by lines~\cite{IM98,HPI00}. Very recently, Chan~\etal{}~\cite{CCGKLZ25} proposed an $\aO\round{T(n) \cdot n^{1/d}}$-time algorithm to compute a spanning path with low average crossing number in general set systems, where $T(n)$ is the time complexity of enumerating elements of a given range.
    \item[Input model. ] Low-crossing matchings and spanning paths have many applications in several domains, where the inputs have different representations. In order to encompass all these setups in one algorithmic framework, we assume that the input set system $(X,\S)$ is given by a list of all elements in $X$ and ranges in $\S$ together with a membership oracle, which tells us, for a given element $x \in X$ and range $S \in \S$ whether $x \in S$ in constant time.
\end{description}



\section{Our Results}

Our main result is a new randomized algorithm that improves the running time for matchings with low crossing numbers for general set systems. 

\smallskip
\begin{theorem} \label{thm:mainVCmatchings}
Given a set system $(X,\S)$ with dual-shatter function $\pi^*_{\S}(k)=O(k^d)$, one can compute a perfect matching of $X$ with crossing number $\aO \left( n^{1-1/d} \right)$ in expected time $\tilde{O} \left( mn + n^2\right)$.
\end{theorem}
\begin{remark*} The parameter $d \geq 1$ need not be an integer.
\end{remark*}
\Cref{thm:mainVCmatchings} gives the following corollary on low-crossing spanning paths (the proof is given in Appendix~\ref{sec:proof_of_corollary_ref_cor_spanpath}).
\begin{corollary} \label{cor:spanpath}
Given a set system $(X,\S)$ with dual-shatter function $\pi^*_{\S}(k)=O(k^d)$, one can compute a spanning path of $X$ with crossing number $\aO \left( n^{1-1/d} \right)$ in expected time $\tilde{O} \left( mn + n^2\right)$.
\end{corollary}
\red{Our algorithm is almost linear in the input size and improves upon the previous-best method for the case $m = o \left( n^{1+2/d} \right)$, which is an important setting for fundamental combinatorial problems, as well as for geometric systems admitting small test-sets. We list some of the applications leading to improved running times in~\Cref{sec:algorithmic_consequences}.}


\red{
Like previous algorithms, we construct the matching iteratively, adding one edge at each iteration. 
What makes the problem computationally hard is that simply picking a uniformly random edge at each iteration fails.
To avoid this problem, earlier algorithms fell into two categories:
a) either they spend a lot of time to carefully choose the edge to add~\cite{CW89}; or b)  they sample an edge according to a weight distribution maintained on \emph{all} $\binom{n}{2}$ edges, which then has to be carefully updated at each iteration~\cite{CsM21}.  Besides being complicated, this becomes prohibitively expensive for computing a partition into $k$-sized sets instead of edges, where one would then need to track the weights of $\Omega \left(n^{k}\right)$ $k$-sized sets.

The new algorithm combines the best of both worlds, with an idea that was surprisingly missed in  earlier works: we maintain weights only on $\S$ and at each iteration, we pick a random range in $\S$ according to these weights. Now, at iteration $t$, we simply pick a \emph{random} cell of the Venn diagram induced by the $t$ previously picked random ranges, and pick a uniform random edge from that cell. This idea allows for a fast and simple algorithm.

This new, \emph{quasi-random}, way of selecting a good edge at each iteration
 avoids the problem of assigning weights to edges, and 
contains just enough non-randomness to guarantee a low-crossing matching.
As a plus, the method can trivially  be adapted to compute a low-crossing partition into sets of size $k$, by picking a random $k$-element subset from the selected cell of the Venn diagram.}
Indeed, \Cref{thm:mainVCmatchings} is  an immediate consequence of a more general algorithm 
for constructing low-crossing balanced partitions. Given a set  $X$ and a parameter $t \in \big[ 2,\lfloor n/2\rfloor \big]$, a balanced $t$-partition of $X$ is a partition of $X$ into $t$ sets $ P_1, \ldots, P_t $ such that
for all $i \in [t]$, we have $|P_i| = \floor{\frac{n}{t}}$ or $|P_i| = \ceil{\frac{n}{t}}$.
We say that a range $S \in \S$ crosses a set $P_i$ if $0 < |S \cap P_i| < |P_i|$ holds. The crossing number of a partition of $X$ is defined to be the maximum number of sets of the partition crossed by some $S \in \S$.
We can derive \Cref{thm:mainVCmatchings} from the following
more general result (the   method \nameref{GDalgo} is presented in Algorithm~\ref{algo:GD}, \Cref{sec:algorithm}).

\begin{restatable}{theorem}{mainthm}  \label{thm:main}
Given a set system $(X,\S)$ with dual-shatter function $\pi^*_{\S}(k)=O(k^d)$, $\proba \in (0,1)$, and a positive integer $t \in \big[ 2,\floor{ n/2} \big]$, the algorithm
\nameref{GDalgo}$\big((X,\S), t,\proba \big)$ returns a balanced $t$-partition  with expected crossing number 
    \[
   O\round{ \max\{t^{1-1/d} , \log t\} \cdot \round{\frac{\ln (mt)}{\proba} + \frac{n}{t} \ln n}},
    \]
    in expected time $O\round{mn\proba + nt + mt}$.
In particular, setting $p = t/n$, we get a balanced $t$-partition with expected crossing number
  \( 
        O\round{\frac{n}{t^{1/d}} \ln \left(mn\right)},
    \)
    in expected time $O\round{t(m + n)}$.
\end{restatable}
\textbf{Remarks.}
\begin{itemize}
\item \Cref{thm:main} immediately implies \Cref{thm:mainVCmatchings} by setting $t =  \floor{n/2}$.
\item \Cref{thm:main} gives balanced $t$-partitions with non-trivial crossing number (that is, crossing number $o(t)$) for  $t = \omega \left( (n\ln(mn))^{d/(d+1)} \right)$.  

\end{itemize}

\subsection{Algorithmic Consequences}\label{sec:algorithmic_consequences}

\Cref{thm:mainVCmatchings} improves upon the previous-best running time of $\tilde{O} \left(n^{2+2/d} +mn^{2/d} \right)$ for the case when $m = o \left( n^{1+2/d} \right)$. 
In this section, we list several consequences, across different areas, which fall in this regime. 
The table below summarizes the main running time improvements.


            \begin{tabular}{clclr}
                \toprule
                \multicolumn{1}{c|}{\textsc{Problem}}
                & \multicolumn{2}{c|}{\textsc{Previous-best running time}}
                & \multicolumn{2}{c}{\textsc{Our result}}
                \\ \midrule
                \multicolumn{1}{c|}{    \begin{tabular}{@{}c@{}} Combinatorial discrepancy \\ dual VC $=d$, $m = \tilde{O}(n)$  \end{tabular}}
                & $\aO \round{n^{2 + \frac{2}{d}}}$
                &\multicolumn{1}{r|}{(SoCG '21)~\cite{CsM21} }
                & \begin{tabular}{@{}c@{}}$\bm{\aO \round{n^2 }}$  \\
                \small{(\Cref{cor:discrepancy})}  \end{tabular}
                 \\ \midrule
                \multicolumn{1}{c|}{    Distance-$1$ graph neighborhood cover}
                &  $\aO \round{n^{4+\eps}}$
                &\multicolumn{1}{r|}{  \small{(FOCS '24)~\cite{DEMMPT23}}  }
                &  \begin{tabular}{@{}c@{}}$\bm{\aO \round{n^2 }}$  \\
                \small{(\Cref{cor:GT-neigh})}   \end{tabular}
                 \\ \midrule
                \multicolumn{1}{c|}{    Matchings  w.r.t. balls in $\RR^d$  }
                & $\aO \round{n^{2 + \frac{3}{d}}}$ 
                &\multicolumn{1}{r|}{(SoCG '21)~\cite{CsM21} }
                & \begin{tabular}{@{}c@{}} $\bm{\aO \round{n^{2 + \frac{1}{d}}}}$ \\
                 \small{(\Cref{cor:balls})} \end{tabular}
                \\ \midrule
                \multicolumn{1}{c|}{  \begin{tabular}{@{}c@{}}  AMM supporting LMSR \\ with $m = \aO(n)$ \end{tabular} }
                & $\tilde{O}(n^{2+2/d})$
                &\multicolumn{1}{r|}{(SODA '25)~\cite{HWY25} }
                & \begin{tabular}{@{}c@{}} $\bm{\tilde{O}(n^2)}$ \\
                 \small{(\Cref{cor:AMM})} \end{tabular}
                 \\
                 \bottomrule
            \end{tabular}%


We give the precise statements of the above table below; the proofs follow known techniques, and are presented in the Appendix.

\begin{description}
\item[Combinatorial discrepancy.]
Using \Cref{thm:mainVCmatchings} in a standard iterative halving scheme \cite{MWW93} gives the \emph{first} quadratic time algorithm for combinatorial discrepancy for set systems with bounded dual shatter functions with optimal guarantees (up to a $\log$ factor).
\smallskip
\begin{corollary}\label{cor:discrepancy}
Let $(X,\S)$, $n = |X|$, $m = |\S|$,  be a set system and  $d$ be a constant such that $\pi^*_{\S}(k) =O\round{ k^d}$. Then there is a randomized algorithm which constructs a coloring $\chi$ of $X$  with expected discrepancy 
   $
   O\round{\sqrt{  n^{1-1/d}\ln n \ln m  + \ln^2 m \log n}}
   $, 
   in expected time
      $\tilde{O} \left( mn + n^2\right)$. 
  \end{corollary}
  \begin{remark*}
    For this setup, the previous-best algorithm with $\aO\round{\sqrt{n^{1-1/d}}}$ discrepancy guarantee had running time $\tilde{O} \left(n^{2+2/d} +mn^{2/d} \right)$~\cite{CsM21}.
  \end{remark*}
\smallskip
\item[Graph theory.]
Given a graph $G = (V, E)$, a distance-$1$ neighborhood cover of $G$ is a family $\K$ of subsets of $V$ such that for each $u \in V$, there
exists a set $F \in \K$ such that $N_G[u] \subseteq F$, where $N_G[u]$ denotes the closed neighborhood of $u$. A key theorem
in a recent breakthrough in FO model checking~\cite{DEMMPT23} (FOCS 2024) stated that if the set system $\left(V, \left\{N_G[u] \colon u \in V\right\}\right)$ has primal shatter function $\pi_\S(k) = O \left( k^{1+\eps} \right)$ for some $\eps > 0$,
then there exists a distance-$1$-neighborhood cover $\K$ of $G$ such that each vertex of $V$ lies in $O \left(n^{\eps}\right)$ sets of $\K$ and each set
of $\K$ has diameter at most $4$.
The proof used the original result of Chazelle and Welzl~\cite{CW89} to show that there exists a spanning path on $V$ such that for each $u \in V$, the set $N_G[u]$ is the union of $O\left( n^{\eps} \right)$
intervals of this path.
Using \Cref{cor:spanpath}, we improve the running time of $O \left(n^{4+\eps} \right)$ of their algorithm to get the following theorem.

\begin{corollary}\label{cor:GT-neigh}
Given $G = (V, E)$ on $n$ vertices with its adjacency list or matrix, if the set system $\left(V, \left\{N_G[u] \colon u \in V\right\}\right)$ has primal shatter function $\pi_\S(k) = O \left( k^{1+\eps} \right)$ for some $\eps > 0$,
then one can construct a distance-$1$ neighborhood cover $\K$ of $G$, such that each vertex of $V$ lies in $\tilde{O} \left(n^{\eps}\right)$ sets of $\K$ and each set
of $\K$ has diameter at most $4$, in expected time $\tilde{O} \left(n^2 \right)$.
\end{corollary}

Combining \Cref{cor:GT-neigh} with the algorithmic framework of \cite{DEMMPT23}, we improve the running time of their main result \cite[Theorem 1]{DEMMPT23} from $O \left(n^{5+\eps}\right)$ to expected time $\tilde{O} \left(n^{3+\eps} \right)$. {\color{myred}We note that for $\eps = 0$ a result similar to \Cref{cor:GT-neigh} was proved by Dreier and Kuske~\cite{DK26}, who also obtained the same running time improvement for \cite[Theorem 1]{DEMMPT23}.}

\item[Geometry.]
Certain geometric families (e.g. half-spaces or balls) admit \emph{test-sets} that are small-sized subfamilies such that if a matching has low crossing number with respect to this subfamily, then it is guaranteed to have low crossing number with respect to the whole family. In particular, Csik\'os and Mustafa observed that any set system $(X,\S)$ induced by balls in $\RR^d$ admits a test-set $\Q_X$ of size $\aO\round{n^{1+1/d}}$ \cite[Lemma 17]{CsM21} (the precise statement is given in the Appendix). Combining this with their $\tilde{O} \left(n^{2+2/d} +mn^{2/d} \right)$ time algorithm, they achieved a running time of $\tilde{O} \left(n^{2+3/d} \right)$ for balls in $\RR^d$. Using the same test-set lemma, that is, applying \Cref{thm:mainVCmatchings} to $(X,\Q_X)$, we obtain the following corollary.
\begin{corollary}\label{cor:balls}
Given a set $P$ of $n$ points in $\RR^d$, 
one can compute a perfect matching of $P$ with crossing number $O \left( n^{1-1/d} \ln n \right)$ in time $\tilde{O} \left( n^{2+1/d} \right)$ for the set system induced by balls.
\end{corollary}

\item[Subquadratic algorithms.]

 \Cref{thm:main} can be combined with a hierarchical method and uniform sampling to get an $o(mn)$-time algorithm for matchings with sub-linear crossing numbers.
\begin{restatable}{corollary}{sublineardelta}
\label{thm:sublinear-delta}
Let $(X,\S)$ be a set system with 
dual-shatter function $\pi^*_{\S}(k)\leq c k^d$, where $c>0$
and $d\geq 1$ are fixed. For every $\delta\in(0,1)$, there is
a randomized algorithm that constructs a perfect matching of $X$
with expected crossing number $ \tilde{O}\!\left(n^{1-1/d+\delta}\right)$
and expected running time $ \tilde{O}\left(
        n^{2-\delta}+n^{1+\delta}
        +m\left(
            n^{1-\delta+\delta/d}
            +n^{1+\delta-(1-\delta)/d}
        \right)
    \right).$
\end{restatable}

In particular, if $m=O(n)$, choosing $\delta=1/(2d)$ gives
expected crossing number $\tilde{O}(n^{1-1/(2d)})$ and
expected running time
    $\tilde{O}\!\left(n^{2-1/(2d)+1/(2d^2)}\right).$
This running time is subquadratic for every fixed $d>1$.


\smallskip

\item[Game theory.] In a prediction market, traders can buy and sell bets (formally called securities) that pay out based on the outcomes of future events. These markets can be modelled with a set system $(X,\S)$, where $X$ denotes the set of possible outcomes and $\S$ contains the \emph{combinations} of outcomes on which traders can place bets. Given a set system and a market state, an automated market maker's (AMM) task is to determine the prices of (combinations of) securities using a cost function and update the prices and market states after each trade. Some popular cost functions are the logarithmic market scoring rule (LMSR) and the quadratic scoring rule (QMSR), whose definitions are given in Appendix~\ref{sec:Apx-AMM}.

Recently, Hossain, Wang, and Yu~\cite{HWY25} (SODA 2025) showed that one can construct an AMM with sub-linear time complexity for a combinatorial prediction market with an LMSR cost function if and only if the corresponding set system has bounded VC-dimension. \\

In the case of bounded VC-dimension, a key algorithmic tool in~\cite{HWY25}  was a data-structure designed to maintain multiplicatively evolving weights on $X$ and to support a query operation, which returns the total weight of a range in $\S$. 
Formally, in the  multiplicative range query and range update (RQRU) problem, given a set system $(X,\S)$ and a weight function $W : X \to \RR_+$, we want to efficiently perform for any $ S \in \S$ and $ w \in \RR_+$: 
\begin{itemize}
    \item $\mathrm{query}(S)$ :  return $\sum\limits_{x \in S} W(x)$;
    \item $\mathrm{update}(S,w)$ : update $W$ with 
     $W(x) =  \begin{cases} w \cdot W(x) \text{ if  } x \in S \\ W(x) \text{ otherwise.}\end{cases}  $
 \end{itemize}
 \smallskip 
Hossain \etal{}~\cite{HWY25} showed that one can construct a data-structure for the multiplicative RQRU problem with sub-linear time complexity per operation, using partition trees based on low-crossing spanning paths~\cite{CW89}. \Cref{thm:main} together with their framework yields the following corollary.
 \begin{corollary}\label{cor:RQRU}
Given a set system $(X,\S)$ with dual shattering dimension $d$, one can construct a data-structure supporting multiplicative RQRU with  $\aO\round{n^{1-1/d}}$ time per operation and with an expected preprocessing time of $\aO(mn+n^2)$.
\end{corollary}

Combining \Cref{cor:RQRU} with the framework of Hossain \etal{}~\cite{HWY25}, we get the following corollary.

\begin{corollary}\label{cor:AMM}
Given a prediction market with associated set system $(X,\S)$ with dual shattering dimension $d$, one can construct an AMM supporting LMSR (or QMSR) cost function with running time $\aO\round{n^{1-1/d}}$ with an expected preprocessing time of $\aO(mn+n^2)$.
\end{corollary}
\begin{remark*}
    For this setup, the previous-best algorithm with $\aO\round{n^{1-1/d}}$ time per operation had preprocessing time $\tilde{O} \left(n^{2+2/d} +mn^{2/d} \right)$~\cite{HWY25}.
  \end{remark*}

We get improved preprocessing times in applied scenarios, with $m < n$, for instance, \textsc{top $L$-candidates} where traders can bet on the outcome of top $L$ candidates among $K \geq L$ candidates. There are a total of $m:= {K \choose L}$ top $L$ securities and $n := K!$ possible outcomes.

\smallskip
\end{description}

\textbf{Organisation. } We present the new algorithm in \Cref{sec:algorithm}. The proof of \Cref{thm:main} is presented in \Cref{proofthm:main} and the proofs of the Corollaries are given in the Appendix.



\section{The Algorithm}
\label{sec:algorithm}

The key idea in the original proof of Welzl and Chazelle-Welzl (and its variations that have been developed in the past 36 years) is to iteratively construct a matching $M$ edge-by-edge, guided by an exponential weight distribution on the ranges. 
\begin{description}
\item[The initial algorithm of Chazelle and Welzl (1989)] assigned a weight to each set of $\S$, where $S \in \S$ had weight that exponentially increased with the edges of the current matching $M$ crossed by $S$.
Then at each iteration, it picked an edge into $M$ that minimized the sum of the weights of the sets crossing that edge.
The analysis crucially used the fact that in every iteration we picked a minimum edge, which turns out to be the algorithmic bottleneck.

\item[Recent work] was able to improve the running times by assigning a weight distribution on the edges~\cite{CsM21}, with the weight essentially based on the  above function of Chazelle-Welzl. Then at each iteration, the new edge is picked by \emph{sampling} from the edges according to this probability distribution.
This simplified the selection of edges at each iteration, which resulted in faster running times.
\end{description}

In this paper, we make the first departure from this general scheme: to further speed up the algorithm, we will no longer track the weight of each individual edge in $\binom{X}{2}$. Instead, 
  we will maintain a partition of $X$ into clusters and maintain (implicitly) a weight distribution on these clusters. Then at each iteration, we sample an edge from a \emph{quasi-uniform} distribution: we first pick a cluster according to this weight distribution and then an edge  spanned by points in that cluster uniformly at random. 
The key point here is that this distribution can be maintained more efficiently than earlier methods.
The clusters  will be adjusted incrementally after each iteration. 
In parallel, we maintain a weight distribution on $\S$, following the classic exponential reweighing scheme with random updates (via martingale analysis), which will guarantee low crossing numbers.

The proof  is via a subtle charging argument: at each iteration, we refine the partition using  a range $S_i \in \S$ picked according to the current weight distribution on $\S$. Thus no edge crossing $S_i$ will be picked in the future. This is then used to argue that the crossing number of the sets due to future edges is small. 


\paragraph*{Algorithm for general partitions}

\begin{algorithm}[t] 
\linespread{1}

\algotitle{\textsc{Partition}}{GDalgo}
\caption{  \textsc{Partition}$\big((X,\S), t,\proba \big)$} \small
    \label{algo:GD}

    $n \leftarrow |X|$, 
    $k \leftarrow 1$

    \While{ $|X| > \floor{ \frac n t}$}{
    
    $\{P^k_1, \dots, P^k_{T_k}\} \leftarrow $ \nameref{subalgo}$\big((X,\S), \floor{ \frac n t}, p \big)$

    $\R^{k} =   \{P^k_1, \dots, P^k_{T_k}\}$

    $X = X \setminus \big(P^k_1 \cup \dots \cup P^k_{T_k}\big)$
    
    $k \leftarrow k + 1$
    }

  $\tau \leftarrow k - 1$

   \Return $\R = \R^1 \cup \cdots \cup R^{\tau} \cup \{X\}$.
\end{algorithm}

 The previously highlighted ideas are formalized in the technical subroutine \nameref{subalgo} (\Cref{algo:subGD}). 
We note that in \nameref{subalgo}, we will not  maintain the distribution $\omega_i$ and its support $\P_i$ explicitly. We give more details on the algorithmic way to maintain and sample from $\omega_i$ efficiently in the next section (proof of \Cref{thm:main}).

\begin{algorithm}[t] 
\linespread{1}

\algotitle{\textsc{SubAlgo}}{subalgo}
\caption{  \textsc{SubAlgo}$\big((X,\S), a, \proba \big)$} \small
    \label{algo:subGD}
    
    $\sigma_1(S) \leftarrow 1$ for all $S\in \S$

    $\P_1 \leftarrow {X \choose a}$ 

    $\omega_1(P) \leftarrow 1$ for all $P \in \P_1$

    $i \leftarrow 1$

    \While{ $\P_i \neq \emptyset$}{

    sample a range $S_i$ according to $\sigma_i$

    sample a part $P_i$ according to $\omega_i$

    $\S_i \leftarrow$ i.i.d. uniform random sample of $\S$ with probability $\proba$

    define $\sigma_{i+1}$ from $\sigma_i$ by doubling the weight  of each range in $\S_i$ crossing $P_i$

    $\P_{i+1} \leftarrow \P_i \setminus \big( \{ \text{parts in } \P_i \text{ crossing } S_i \} \cup \{ \text{parts in } \P_i \text{ intersecting } P_i\} \big)$

    $\omega_{i+1}$ $\leftarrow$ uniform distribution on $\P_{i+1}$

    $i \leftarrow i+1$

    }


    \Return $ \{P_1, \dots, P_T\}$
\end{algorithm}

As we stated in the introduction, our main result on matchings with low crossing numbers is a special case of \Cref{thm:main} that concerns low-crossing balanced $t$-partitions.

To avoid unnecessary repetitions, we present the algorithm for general $t$-partitions, from which the algorithm for low-crossing matchings can easily be deduced via substituting $t = \floor{n/2}$. Furthermore, to keep the description simple, the algorithm \nameref{GDalgo} constructs (possibly more than $t$) parts of size \(\lfloor n/t\rfloor\) and returns all the remaining points as the last part. To get exactly \(t\) balanced parts, one can extend the algorithm to \(\mathcal P=\binom{X}{\lfloor n/t\rfloor}\cup\binom{X}{\lceil n/t\rceil}\). If \(n=t\floor{n/t}+r\), with \(0\le r<t\), the adapted algorithm would first build \(r\) parts of size \(\ceil{n/t}\) and then \(t-r\) parts of size \(\floor{n/t}\), stopping each phase once the required number of parts has been built.


The main algorithm \nameref{GDalgo} 
that is used to establish \Cref{thm:main} is presented in \Cref{algo:GD}. Here $p \in (0, 1)$ is a user-specified parameter that dictates the trade-off between the running time and the crossing number.

A challenge that comes with this approach is that, as the result of the refinement at each iteration,   it may be that after $\Theta \left(t^{1/d}\right)$ iterations, each cluster of $X$ contains only $<n/t$ points (this follows from the dual shatter function bound). Thus we can no longer continue, and are forced to restart with the remaining points all in one cluster. However, we show that the losses due to this re-start form   a geometric series which then gives the required bound at the end.




\section{Proof of Theorem \ref{thm:main} } \label{proofthm:main}


\Cref{thm:main} will be based on three key properties of \nameref{algo:subGD}, which we now state. The proof of these properties will be presented in later sections.

Let $\P = \binom{X}{  a }$, and define the function $\Inc : \S \times \P \to \{0,1\}$ as
\[
    \Inc(S,P) = \begin{cases} 1 \quad &\text{ if } S \text{ crosses } P, \\ 0 &\text{ otherwise.} \end{cases}
\]

The first property follows from the update rule of the weight functions $\sigma_i$'s.

\begin{restatable}{property}{lemmaMWUonranges}\label{lemma:MWU-on-ranges}
Let $T$ be the stopping time of \nameref{subalgo}$\big((X,\S), a, \proba \big)$, with the returned parts $P_1, \dots, P_T$. Then, 
\begin{equation}\label{eq:MWU-UB}
   \EE\bracket{ \max_{S \in \S} \sum\limits_{i=1}^T \Inc(S,P_i)}
    \leq  
    \frac{2\log m +4\ln(mT)+2}{\proba} +  \frac{2}{\ln 2} \cdot \sum\limits_{i=1}^T ~ \EE\bracket{ \sum\limits_{S \in \S} \frac{\sigma_i(S)}{\sigma_i(\S)} \cdot \Inc (S,P_i)} \, .  
\end{equation}
\end{restatable}
Note that \Cref{lemma:MWU-on-ranges} in fact provides an upper-bound on the overall crossing number of parts returned by a single call of \nameref{subalgo}. 

By analysing the evolution of weights $\omega_i$, we get the next property.
\begin{restatable}{property}{lemmatotalsum}\label{lemma:totalsum}
Let $T$ be the stopping time of \nameref{subalgo}$\big((X,\S), a, \proba \big)$ and let $S_1, \dots, S_T$ be the ranges picked by \nameref{subalgo}. Then we have 
\[
   \sum\limits_{i=1}^T \sum\limits_{P \in \P} \frac{\omega_i(P)}{\omega_i(\P)} \Inc(S_i,P) 
   = 
   O\round{\log{|X| \choose a}} 
   = 
   O\round{a\log{|X|}}.
\]
\end{restatable}

The next property relates the above two in expectation.

\begin{restatable}{property}{lemmadoublecount}\label{lemma:double-count}
Let $P_1, \dots, P_T$ and $S_1, \dots, S_T$ be the random variables corresponding to the parts and ranges picked by \nameref{subalgo}$\big((X,\S), a, \proba \big)$.
Then for each $i \in [T]$, we have 
\[
\EE \bracket{  \sum\limits_{S \in \S} \frac{\sigma_i(S)}{\sigma_i(\S)} \cdot \Inc (S,P_i) }
=
\EE\bracket{   \sum\limits_{P \in \P} \frac{\omega_i(P)}{\omega_i(\P)} \Inc(S_i,P)}, 
\]
where the above expectations are over the choices of $S_1, \ldots, S_i, P_1, \ldots, P_i, \S_1, \ldots, \S_i$. 
\end{restatable}

The above three properties are consequences of the weight management of the algorithm. In particular, they hold for any set system. The key point where we use the structural assumption 
on the set system (bounded dual shatter dimension) is in analysing the halting time of \nameref{subalgo}.

\begin{restatable}{lemma}{lemmahaltingtime}\label{lemma:haltingtime}
For any set system $(X,\S)$ with dual-shatter function $\pi^*_{\S}(k) = c\cdot k^d$ and an integer $a \in [1,|X|/2]$, the halting time $T$ of \nameref{subalgo}$\big((X,\S), a, p \big)$ satisfies
\[
T > \round{\frac{|X|}{a(c+1)}}^{1/d}.
\]
That is, the union of parts returned by \nameref{subalgo}$\big((X,\S), a,p \big)$ covers $\Omega \round{a \cdot \round{\frac{|X|}{a}}^{1/d}}$ elements of $X$.
\end{restatable}
\begin{proof}[Proof of Lemma~\ref{lemma:haltingtime}]
The algorithm \nameref{algo:subGD} continues as long as $\P_i \neq \emptyset$. In the beginning, $\P_1$ has $\binom{|X|}{a}$ elements, and at iteration $i$ we remove the parts in $\P_i$ that either
\begin{itemize}
    \item  intersect $P_i$, or
    \item  are crossed by the range $S_i$.
\end{itemize}
As no element in $P_i$ can appear in any later part, we consider the elements in $P_i$ \emph{covered} or removed,  for all later iterations.

This means that after $\ell$ iterations, $\P_{\ell+1}$ contains those $a$-element subsets of the uncovered set of  $|X| - \ell\cdot a$ elements that are not crossed by $S_1, \dots, S_\ell$. Observe that an $a$-subset is not crossed by  $S_1, \dots, S_\ell$ if and only if all of its elements belong to the same equivalence class w.r.t. $S_1, \dots, S_\ell$. Therefore, it is sufficient to guarantee that at least one of these equivalence classes has at least $a$ elements which then implies that $\P_{\ell+1} \neq \emptyset$. 

Since $(X,\S)$ has dual-shatter dimension $d$, the ranges $S_1, \dots, S_\ell$ define at most $c\cdot \ell^d$ equivalence classes on $X$. By the pigeonhole principle, there exists an equivalence class with at least
\(
\frac{|X| - \ell a}{c\cdot \ell^d} 
\)
elements. Thus, as long as we have
\(
\ell
\leq 
\round{\frac{|X|}{a(c+1)}}^{1/d},
\)
there will be an equivalence class with at least 
\[
\frac{|X| - \ell a}{c\cdot \ell^d} 
\geq 
\frac{|X| - \ell^d a}{c\cdot \ell^d} 
\geq
\frac{|X| - \frac{|X|}{a(c+1)}\cdot a}{c\cdot \frac{|X|}{a(c+1)}} 
=
a
\]
elements and thus the algorithm \nameref{algo:subGD} continues. In conclusion, the stopping time $T$ has to satisfy
\[
T > \round{\frac{|X|}{a(c+1)}}^{1/d}.
\]
\end{proof}



Now we are ready to establish our main result -- \Cref{thm:main}.

Assume that  \nameref{GDalgo} made $\tau$ calls to the subroutine \nameref{subalgo}, resulting in the sub-partitions $\R^1, \dots, \R^{\tau}$,
and let $T_k = |\R^k|$ be the number of parts returned by the $k$-th run  of \nameref{subalgo}. Further, 
let $\R^k = \{P_1^k, \dots, P_{T_k}^k\}$ be the parts and  $\{S_1^k, \dots, S_{T_k}^k\}$ be the ranges picked during the $k$-th run.\\

\begin{description}
\item[Bounding subroutine crossing number.] ~\\
By Properties \ref{lemma:MWU-on-ranges}, \ref{lemma:double-count}, and \ref{lemma:totalsum} and substituting $a = \floor{n/t}$, we have for any $k \in [1,\tau]$,
\begin{align*}
 &\EE\bracket{\max\limits_{S \in \S} \sum\limits_{i=1}^{T_k} \Inc (S,P^k_i) }
 \leq\\
&\leq
\frac{2\log m +4\ln(mT_k)+2}{\proba} +  \frac{2}{\ln 2}\cdot  \sum\limits_{i=1}^{T_k} ~\EE\bracket{  \sum\limits_{S \in \S} \frac{\sigma_i(S)}{\sigma_i(\S)} \cdot \Inc (S,P^k_i)} \\
& =
\frac{2\log m +4\ln(mT_k)+2}{\proba} +  \frac{2}{\ln 2} \cdot\sum\limits_{i=1}^{T_k} ~ \EE\bracket{ \sum\limits_{P \in \P} \frac{\omega_i(P)}{\omega_i(\P)} \Inc(S^k_i,P)} \\
& =
O\round{\frac{\log (mt)}{\proba} + a \log n}
=
O\round{\frac{\log  (mt)}{\proba} + \frac{n \log n}{t}}.\numberthis \label{eq:all-together}
\end{align*}

\item[Bounding $\tau$.]
We count the number of \nameref{subalgo} calls made by \nameref{GDalgo} in $\log n$ phases, where phase $j$ lasts as long as  $n/2^j < |X| \leq n/2^{j-1}$. Note that 
as $|X| \geq \floor{n/t}$, there are at most $\log t$ phases. 
Let $h_j$ denote the number of calls of \nameref{subalgo} in phase $j$. 
By Lemma \ref{lemma:haltingtime}, each call in phase $j$ outputs a sub-partition covering at least 
\[
\frac{(n/2^j)^{1/d}\floor{n/t}^{1-1/d}}{(c+1)^{1/d}}
\]
elements. Thus, $h_j$ can be bounded as
\begin{align*}
h_j
\leq
\frac{n/2^j}{\frac{(n/2^j)^{1/d}\floor{n/t}^{1-1/d}}{(c+1)^{1/d}}}
=
(c+1)^{1/d} \cdot \round{\frac{n/2^j}{n/t}}^{1-1/d}
=
(c+1)^{1/d} \cdot \round{\frac{t}{2^j}}^{1-1/d}.
\end{align*}
This gives the following bound on the total number of calls $\tau$ to \nameref{subalgo}:
\begin{equation} \label{eq:taubounding}
\tau = \sum\limits_{j=1}^{\log t} h_j
\leq 
\sum\limits_{j=1}^{\log t} (c+1)^{1/d} \cdot \round{\frac{t}{2^j}}^{1-1/d}
= 
\begin{cases}
O(t^{1-1/d})\quad &\text{ if } d >1 \\
O(\log t)\quad &\text{ if } d =1
\end{cases} \quad .
\end{equation}

\item[Bounding total crossing number.]
Note that the crossing number of the partition $\R$ returned by \nameref{algo:GD} can be bounded as
\begin{equation*}\label{eq:max-per-subroutine}
    \max\limits_{S \in \S} \sum\limits_{P \in \R} \Inc (S,P)  
    \leq 
    \max\limits_{S \in \S} \sum\limits_{k=1}^{\tau} \sum\limits_{i=1}^{T_k} \Inc (S,P^k_i) +1
    \leq
    \sum\limits_{k=1}^{\tau} \max\limits_{S \in \S} \sum\limits_{i=1}^{T_k} \Inc (S,P^k_i)+1, 
\end{equation*}
where the `$+1$' term corresponds to the last part added to $\R$.
Combining this with \Cref{eq:all-together}, we get  
$$
\EE\bracket{\max\limits_{S \in \S} \sum\limits_{P \in \R} \Inc (S,P)   } 
 \leq 
 \sum\limits_{k=1}^{\tau} \EE\bracket{ \max\limits_{S \in \S} \sum\limits_{i=1}^{T_k} \Inc (S,P^k_i) } {+}1
= O\round{  \tau \left( \frac{\log  (mt)}{\proba} + \frac{n \ln n}{t}\right)}. 
$$
 
\item[Running time.]
Finally, we bound the running time of the algorithm \nameref{algo:GD}. It makes $\tau = O\round{t^{1-1/d}}$ calls to \nameref{subalgo} and apart from that, it spends $O(n)$ time in total to update $X$ between the calls.

A key point  is that in \nameref{subalgo}, we do not store  $\P_i$ and $\omega_i$ explicitly. Instead, we  only maintain the equivalence classes on $X$ defined by $\left\{S_1, \ldots, S_{i-1}\right\}$,  that we call \emph{cells} for brevity. Each cell is stored in the form of a list $C$ of elements contained in the cell.
To obtain $\P_{i+1}$, we go through the points of each list, and partition into two lists if $S_i$ crosses that cell. During the same pass, we also delete the elements contained in $P_i$. This takes time $O(n)$ at each iteration, assuming $O(1)$ membership access.

Sampling a range $S_i$ from $\sigma_i$ can be easily done in time $O(m)$. 

Given the lists $C_1, \ldots, C_{c_i}$ at the start of iteration $i$, we can sample $P_i$ from $\omega_i$ in two steps. First, we choose a cell randomly according to the distribution 
\[
\PP\bracket{C_j \text{ is picked}} =\begin{cases}\frac{ \binom{|C_j|}{\floor{n/t}}}{ \sum_{j'} \binom{|C_{j'}|}{\floor{n/t}} } \quad &\text{ if } |C_j| \geq \floor{n/t} \\ 0 &\text{ if }|C_j| < \floor{n/t}. \end{cases}
\]   
Second, we   sample an $\floor{n/t}$-sized subset from the chosen list uniformly at random. Note that then a part of size $\floor{n/t}$ belonging to $\P_i$, say lying in a cell $C'$, is picked with probability
\[ 
\frac{1}{\binom{|C'|}{\floor{n/t}}}\cdot \frac{ \binom{\, |C'|}{\floor{n/t}}}{ \sum_{j'} \binom{|C_{j'}|}{\floor{n/t}} }
=
\frac{1}{ \sum_{j'} \binom{|C_{j'}|}{\floor{n/t}} }, 
\]    
i.e., the uniform distribution on $\P_{i}$. 

The set of ranges $\S_i$ can be computed in time $O(m)$. We have 
$\EE \left[ |\S_i| \right] = m\proba$, and the weights of the ranges in $\S_i$ can thus be updated in expected time $O\left( m\proba \cdot \frac{n}{t}\right)$. 
Thus the total expected overall time is $O\left(n + m +  m\proba \cdot \frac{n}{t} \right)$ for each iteration. 

The total number of iterations throughout the whole run of \nameref{algo:GD} is  $t$, therefore the
expected time complexity of \nameref{algo:GD} is 
\[
O\round{n + t\left(n + m +  m\proba \cdot \frac{n}{t} \right)} = O\round{nt+mt +  mnp}. 
\]\qed
\end{description}

In the remaining sections, we prove Properties \ref{lemma:MWU-on-ranges} and \ref{lemma:totalsum}. The proof of \Cref{lemma:double-count} is a standard double-counting argument, often used in the analysis of MWU algorithms, therefore it is deferred to Appendix \ref{sec:appendix-ommited}.

\subsection{Proof of Property \ref{lemma:MWU-on-ranges}.}\label{sec:lemma:MWU-on-ranges}

\lemmaMWUonranges*

\textbf{Proof.}
For a random event $\E$, let $\mathbf{1}_{\E}$ denote the indicator random variable of $\E$.

We prove the property by analysing the total weight of ranges of $\S$ after $T$ iterations. On one hand, the total weight can be lower-bounded by the maximum:
\[
\sigma_{T+1}(\S) \geq \max\limits_{S \in \S} \sigma_{T+1}(S) = 2^{\max\limits_{S \in \S}\sum\limits_{i=1}^T \Inc(S,P_i)~\cdot~ \mathbf{1}_{ \{S\in \S_i\}}}.
\]
On the other hand, we can examine the evolution of weights over the iterations as follows:
\begin{align*}
\sigma_{T+1}(\S) 
&= 
\sigma_{T}(\S) + \sum\limits_{S \in \S} \sigma_{T}(S) \cdot \Inc(S,P_{T}) ~\cdot~ \mathbf{1}_{ \{S\in \S_{T}\}}\\
&=
 \sigma_{T}(\S)\round{1 +  \sum\limits_{S \in \S} \frac{\sigma_{T}(S)}{\sigma_{T}(\S)} \cdot \Inc(S,P_{T})~\cdot~ \mathbf{1}_{ \{S\in \S_{T}\}}}\\
&
\hspace*{1.0in} \vdots  \\
&  =
\sigma_1(\S)\cdot  \prod\limits_{i=1}^{T} \round{1 +  \sum\limits_{S \in \S} \frac{\sigma_{i}(S)}{\sigma_{i}(\S)} \cdot \Inc(S,P_{i})~\cdot~ \mathbf{1}_{ \{S\in \S_i\}} } \\
& \leq
m \cdot \exp\round{ \sum\limits_{i=1}^{T}\sum\limits_{S \in \S} \frac{\sigma_{i}(S)}{\sigma_{i}(\S)} \cdot \Inc(S,P_{i}) ~\cdot~ \mathbf{1}_{ \{S\in \S_i\}}}.
\end{align*}

Putting the two bounds together and taking the logarithm, we get
\begin{equation}\label{eq:wts-before-EE}
\max\limits_{S \in \S} \sum\limits_{i=1}^T \Inc(S,P_i) ~\cdot~ \mathbf{1}_{ \{S\in \S_i\}}
\leq 
\log m +  \frac{1}{\ln 2} \sum\limits_{i=1}^{T} \sum\limits_{S \in \S} \frac{\sigma_{i}(S)}{\sigma_{i}(\S)} \cdot \Inc(S,P_{i})~\cdot~ \mathbf{1}_{ \{S\in \S_i\}}.
\end{equation}

If $\proba = 1$, then $\mathbf{1}_{\{S \in \S_i\}} = 1$ for all $i$ and  $S\in \S$, thus taking total expectation we obtain
\[
\EE\bracket{
\max\limits_{S \in \S} \sum\limits_{i=1}^T \Inc(S,P_i) 
}
\leq 
\log m +  \frac{1}{\ln 2} 
\EE\bracket{
\sum\limits_{i=1}^{T} \sum\limits_{S \in \S} \frac{\sigma_{i}(S)}{\sigma_{i}(\S)} \cdot \Inc(S,P_{i})
},
\]
which implies \Cref{lemma:MWU-on-ranges}. 
Therefore let us assume that $\proba < 1$.
     
Observe that  for each fixed iteration $i$, the random variables $\{\sigma_i,P_i\}$ and $\S_i$ are independent, thus the expected value of the RHS of Equation~\eqref{eq:wts-before-EE} can be written as
    \begin{align}
        \log m &+ \frac{1}{\ln 2}\sum\limits_{i=1}^{T} \sum\limits_{S \in \S} \EE\bracket{\frac{\sigma_{i} (S)}{\sigma_{i}(\S)}   \Inc(S,P_i)}\EE\bracket{ \bm{1}_{\{S \in \S_i\}}} \nonumber \\ 
        &=
         \log m + \frac{\proba}{\ln 2} \cdot\sum\limits_{i=1}^{T} \sum\limits_{S \in \S} \EE\bracket{\frac{\sigma_{i} (S)}{\sigma_{i}(\S)}   \Inc(S,P_i)}. \label{eq:probaS-linearity-of-exp}
    \end{align}
Separating the two factors in the expectation of the LHS of Equation~\eqref{eq:wts-before-EE} is less straightforward. Our goal is to show that $ \EE\bracket{
\max\limits_{S \in \S} \sum\limits_{i=1}^T \Inc(S,P_i) \cdot \proba
}$ is close to $\EE\bracket{
\max\limits_{S \in \S} \sum\limits_{i=1}^T \Inc(S,P_i) ~\cdot~ \mathbf{1}_{ \{S\in \S_i\}}
}$. However, since we have a maximization, we cannot use linearity of expectation and so we need to be careful with the cross-iteration dependencies (in particular, $\{\sigma_i,P_i\}$ depend on $\S_1, \dots, \S_{i-1}$).
Instead, we observe that we only need to show that 
\begin{equation}\label{eq:1q-UB}
 \EE\bracket{
\max\limits_{S \in \S} \sum\limits_{i=1}^T \Inc(S,P_i) \cdot \frac{\proba}{2}
} 
\leq 
 \EE\bracket{
\max\limits_{S \in \S} \sum\limits_{i=1}^T \Inc(S,P_i) ~\cdot~ \mathbf{1}_{ \{S\in \S_i\}}
}   
+ 2\ln(mT)+1.
\end{equation}
Indeed, using this bound, Equations~\eqref{eq:wts-before-EE} and \eqref{eq:1q-UB} give the following guarantee in expectation:
\begin{align*}
\frac{\proba}{2} \cdot \EE\bracket{
\max\limits_{S \in \S} \sum\limits_{i=1}^T \Inc(S,P_i) } 
&\leq 
 \EE\bracket{
\max\limits_{S \in \S} \sum\limits_{i=1}^T \Inc(S,P_i) ~\cdot~ \mathbf{1}_{ \{S\in \S_i\}}
}   
+ 2\ln(mT)+1\\
&\leq 
\log m + \frac{\proba}{\ln 2} \cdot\sum\limits_{i=1}^{T} \sum\limits_{S \in \S} \EE\bracket{\frac{\sigma_{i} (S)}{\sigma_{i}(\S)}   \Inc(S,P_i)} +2\ln(mT)+1,
\end{align*}
which gives the statement of \Cref{lemma:MWU-on-ranges}.


 Hence, all that remains is to prove that Equation~\eqref{eq:1q-UB} holds.
Observe that for any choice of functions $f,g$, we have $\max f(x) - \max g(x) \leq \max (f(x) - g(x))$, which gives
\begin{align}
\EE&\bracket{
\max\limits_{S \in \S} \sum\limits_{i=1}^T \Inc(S,P_i) \cdot \frac\proba 2
 - 
\max\limits_{S \in \S} \sum\limits_{i=1}^T \Inc(S,P_i)\cdot \mathbf{1}_{ \{S\in \S_i\}} } \nonumber \\
&\leq
\EE\bracket{
  \max\limits_{S \in \S} \sum\limits_{i=1}^T \Inc(S,P_i) \cdot\round{\frac\proba 2 - \mathbf{1}_{ \{S\in \S_i\}}} \label{eq:maxmax}
} \,.
\end{align}
We will prove the following claim about the probability of the RHS random variable exceeding $2\ln(mT)$. Let $U=\floor{|X|/a}$ be a fixed upper bound on $T$.

\begin{restatable}{claim}{martingaleclaim}\label{claim:martingale-bound}
\begin{equation}
    \PP\bracket{
    \max_{S\in\S}\sum_{i=1}^{T}
    \Inc(S,P_i)
    \round{\frac{\proba}{2}-\bm 1_{\{S\in\S_i\}}}
    \geq 2\ln(m{\color{myred}U})}
    \leq \frac{1}{{\color{myred}U}}.
\end{equation}
\end{restatable}

The proof of \Cref{claim:martingale-bound} is a standard application of an Azuma-type inequality for martingales, its proof is deferred to Appendix \ref{sec:appendix-ommited}. 

Applying \Cref{claim:martingale-bound} and using that
$\abs{\sum_{i=1}^{T} \Inc(S,P_i)\cdot
\round{\frac{\proba}{2} - \bm 1_{\{S \in \S_i\}}}}
\leq T{\color{myred}\leq U}$, we get
\begin{align*}
&\EE\bracket{\underbrace{
\max\limits_{S \in \S} \sum\limits_{i=1}^{T}
\Inc(S,P_i)\cdot
\round{\frac{\proba}{2} - \bm 1_{\{S \in \S_i\}}}
}_{\mathcal M}}\\
&{\color{myred}=
\EE\bracket{\mathcal M\cdot
\bm 1_{\{\mathcal M\leq 2\ln(mU)\}}}
+
\EE\bracket{\mathcal M\cdot
\bm 1_{\{\mathcal M>2\ln(mU)\}}}}\\
&\leq
2\ln(m{\color{myred}U})\cdot
\PP\bracket{\mathcal M\leq 2\ln(m{\color{myred}U})}
+
{\color{myred}U}\cdot
\PP\bracket{\mathcal M>2\ln(m{\color{myred}U})}\\
&{\color{myred}\leq
2\ln(mU)+U\cdot\frac{1}{U}
=2\ln(mU)+1.}
\end{align*}
For each fixed range $S$, consider the exponential of half
its sum over iterations $1,\ldots,T$, setting all increments
after iteration $T$ to zero.
Conditional on the history just before $\S_i$ is sampled,
a crossing iteration contributes a factor with mean
$e^{\proba/4}(1-\proba+\proba e^{-1/2})
\leq e^{\proba/4}(1-\proba/4)\leq1$;
otherwise, including after stopping, the factor is $1$.
Applying conditional expectations successively over the fixed
$U$ steps shows that each range's expected exponential is at
most $1$.
Since the maximum of these exponentials is at most their sum,
$\EE\bracket{\exp(\mathcal M/2)}\leq m$.
Jensen's inequality and $T\geq1$ therefore give

$$
\EE\bracket{\mathcal M}
\leq 2\ln\round{\EE\bracket{\exp(\mathcal M/2)}}
\leq 2\ln m
\leq \EE\bracket{2\ln(mT)}+1.
$$

This, together with Equation~\eqref{eq:maxmax}, proves
Equation~\eqref{eq:1q-UB}.


This concludes the proof of \Cref{lemma:MWU-on-ranges}.  
\qed

\newpage

\subsection{Proof of Property \ref{lemma:totalsum}.}\label{sec:lemma:totalsum}

\lemmatotalsum*

\textbf{Proof.}
Recall that $\omega_1, \omega_2, \dots, \omega_T$ are uniform distributions on decreasing supports. In particular, there are sets of parts $\P =  \P_1 \supseteq \dots \supseteq \P_T$ such that 
\[
\omega_i(P) 
= 
\begin{cases} 1 \quad &\text{ if } P \in \P_i \\ 0 &\text{ if } P \in \P \setminus \P_i. \end{cases}
\]
Thus we can express the sum as follows:
\[
\sum\limits_{i=1}^T \sum\limits_{P \in \P}    \frac{\omega_{i} (P)}{\omega_{i}(\P)}   \Inc(S_i,P) 
=
\sum\limits_{i=1}^T \frac{\abs{\curly{P \in \P_i : S_i \text{ crosses } P}}}{|\P_i|}.
\]
Moreover, $\P_i \setminus \P_{i+1}$ consists of those elements (parts) in $\P_i$ that are  
\begin{itemize}
    \item crossed by the range $S_i$, or
    \item intersected by the part $P_i$.
\end{itemize}
Let $c_i$ denote the number of elements in $\P_i$ crossed by $S_i$, and let $m_i = \abs{\P_i \setminus \P_{i+1}}$. Then, $c_i \leq m_i$ and
\[
\abs{\P_{i+1}}
=
\abs{\P_{i}} - m_i
= \dots =
\abs{\P_{1}} - \sum\limits_{j=1}^i m_j. 
\]
Thus, 
\[
\sum\limits_{i=1}^T \frac{\abs{\curly{P \in \P_i : S_i \text{ crosses } P}}}{|\P_i|}
=
\sum\limits_{i=1}^T \frac{c_i}{\abs{\P_{1}} - \sum\limits_{j=1}^{i-1} m_j}
\leq 
\sum\limits_{i=1}^T \frac{m_i}{\abs{\P_{1}} - \sum\limits_{j=1}^{i-1} m_j}\, .
\]
The proof of \Cref{lemma:totalsum} follows immediately from the next claim by substituting $A_i = m_i$ and $N = |\P_1| = \binom{|X|}{a}$.
\begin{claim}
For any $N \in \NN$ and any sequence $\round{A_i}_{i=1}^\infty$ such that $A_i \in \NN$ and $A_i \leq N - \sum\limits_{j=1}^{i-1}A_j$, we have
\[
\sum_{i=1}^\infty \frac{A_i}{N-\sum\limits_{j=1}^{i-1} A_j} = O\round{\log N}\,.
\]
\end{claim}
\begin{proof}
Observe that since $A_i \in \NN$ and $A_i \leq N - \sum\limits_{j=1}^{i-1}A_j$, the series $\round{A_i}_{i=1}^\infty$ has at most $N$ positive terms. Let $k \leq N$ be the number of positive terms and assume wlog that $A_1, \dots, A_k > 0$. Then we have that 
\[
\sum_{i=1}^\infty \frac{A_i}{N-\sum\limits_{j=1}^{i-1} A_j} 
= 
\sum_{i=1}^k \frac{A_i}{N-\sum\limits_{j=1}^{i-1} A_j} \, .
\]
Furthermore, observe that for any $i \in [1,k]$,
\begin{align*}
\frac{A_i}{N-\sum\limits_{j=1}^{i-1} A_j}   
&=
\underbrace{\frac{1}{N-\sum\limits_{j=1}^{i-1} A_j}+ \dots + \frac{1}{N-\sum\limits_{j=1}^{i-1} A_j} }_{A_i \text{ times }}\\
&\leq
\frac{1}{N-\sum\limits_{j=1}^{i-1} A_j}+\frac{1}{N-\round{ 1+\sum\limits_{j=1}^{i-1} A_j}}+ \dots + \frac{1}{N-\round{A_i -1 + \sum\limits_{j=1}^{i-1} A_j}} \\
& = \sum_{b=0}^{A_i-1} \frac{1}{N- b - \sum\limits_{j=1}^{i-1} A_j}. 
\end{align*}
Note that for each $b \in [1, N]$, the term $\frac{1}{N - b}$ appears precisely once over all $k$, and so
\begin{align*}
\sum_{i=1}^k \frac{A_i}{N-\sum\limits_{j=1}^{i-1} A_j} &\leq \sum_{i=1}^k \sum_{b=0}^{A_i-1} \frac{1}{N- b - \sum\limits_{j=1}^{i-1} A_j}
 =  \sum_{b=0}^{ \round{\sum_{i=1}^k A_i} - 1} \frac{1}{N- b}\\
 &\leq \sum_{b=0}^{N-1} \frac{1}{N-b} = O \left( \log N \right). 
\end{align*}
\end{proof}
This concludes the proof of \Cref{lemma:totalsum}. 
\qed

\bibliographystyle{plain}
\bibliography{spanningtreealgo}

\appendix

\section{Proofs of \Cref{claim:martingale-bound} and Property \ref{lemma:double-count}.}\label{sec:appendix-ommited}

\martingaleclaim*

\begin{proof}

We use the following Azuma-type inequality for martingales.

    \begin{lemma}[{\cite[Lemma 10]{KY14}}]\label{lemma:Azuma}
        Let $X = \sum_{i=1}^r x_i$ and $Y = \sum_{i=1}^r y_i$  be sums of non-negative random variables, where $r$ is a random stopping time with finite expectation, $|x_i - y_i | \leq 1$ for all $i$, and
        \[
            \EE \bracket{ x_i - y_i ~\bigg|~ \sum_{s < i} x_s, \sum_{s < i} y_s} \leq 0.
        \]
        Then for any $\eps \in [0,1]$ and $A \in \RR$, we have 
        \(
            \PP\bracket{ (1-\eps) X -Y \geq A } \leq \exp (-\eps A).
        \)
    \end{lemma}
    We apply \Cref{lemma:Azuma} in the following way. 
    For each $i\in [1, T]$ and $S \in \S$, consider the random variables 
\[x_i(S) = \Inc(S,P_i)\cdot \proba \qquad \text{and} \qquad y_i(S) = \Inc(S,P_i)\cdot\bm 1_{\{S \in \S_i\}}, \]
 which are measurable with respect to $P_i$ and $\S_i$. For any $i$ and $S \in \S$, we have $|x_i(S) - y_i(S) | \leq 1$.
Furthermore, since $\S_i$ is independent of $\curly{P_i,\, \sum_{k < i} x_k(S),\,\sum_{k < i} y_k(S)}$, we have
            \[
                \EE\bracket{x_i(S) - y_i(S) ~\bigg|~ \sum_{k < i} x_k(S), \sum_{k < i} y_k(S)} = 0,
            \]
        as $\EE\bracket{\proba - \bm 1_{\{S \in \S_i\}}} =0$ for any $i \in [1,T]$ and $S \in \S$.

        \noindent
        Therefore, \Cref{lemma:Azuma} with $\eps = 1/2$ implies that for any $A \in \RR$ and $S \in \S$, 
\begin{align*}
                 \PP\bracket{\frac{1}{2} \sum\limits_{i=1}^{T}  \Inc(S,P_i) \cdot \proba  - \sum\limits_{i=1}^{T}  \Inc(S,P_i) \cdot  \bm 1_{\{S \in \S_i\}}  \geq A} \leq \exp\left(-\frac{A}{2}\right).
\end{align*}
Rearranging the terms and applying the union bound over all $S \in \S$, we get 
            \begin{align*}
                 \PP\bracket{\max\limits_{S \in \S} \sum\limits_{i=1}^{T}  \Inc(S,P_i)  \cdot\round{\frac \proba 2 - \bm 1_{\{S \in \S_i\}} } \geq A} \leq |\S|\exp\left(-\frac{A}{2}\right).
            \end{align*}
            Setting $A =  2\ln(|\S|U)$, we conclude the proof of \Cref{claim:martingale-bound}.
   \end{proof}

\lemmadoublecount*

\textbf{Proof.}
     The proof will follow by conditioning on the past choices of the algorithm. Let  $F_{i}$ be one realization of all random variables used by the algorithm in iterations $j < i$; that is,  $ F_{i} =\round{P_1,\dots, P_{i-1}, S_1, \dots, S_{i-1}, \S_1, \dots, \S_{i-1}}$ and let $\F_i$ denote the set of all such feasible tuples. We have
\begin{align*}
    &\EE \bracket{ \sum\limits_{S \in \S} \frac{\sigma_{i} (S)}{\sigma_{i}(\S)}   \Inc(S,P_i) }
    =
    \sum\limits_{\phi \in \F_{i}}\EE_{P_i \sim \omega_i} \bracket{ \sum\limits_{S \in \S} \frac{\sigma_{i} (S)}{\sigma_{i}(\S)}    \Inc (S,P_i) ~\bigg|~ F_{i}= \phi}  \cdot \PP \bracket{F_i = \phi},  \text{ and}\\
    &\EE \bracket{\sum\limits_{P \in \P} \frac{\omega_{i} (P)}{\omega_{i}(\P)}   \Inc(S_i,P) }
    =
    \sum\limits_{\phi \in \F_{i}}\EE_{S_i \sim \sigma_i} \bracket{ \sum\limits_{P \in \P} \frac{\omega_{i} (P)}{\omega_{i}(\P)}   \Inc(S_i,P) ~\bigg|~ F_{i}= \phi}  \cdot \PP \bracket{F_i = \phi}.
\end{align*}
  Observe that $\omega_i$ and $\sigma_i$ are measurable with respect to $F_{i}$, and thus for any $\phi \in \F_i$, we have
\begin{align*}
        &\EE_{P_i \sim \omega_i}\bracket{ \sum\limits_{S \in \S} \frac{\sigma_{i} (S)}{\sigma_{i}(\S)}  \cdot \Inc(S,P_i)~\bigg|~ F_{i}= \phi}
        =
        \sum\limits_{P \in \P} \frac{\omega_{i} (P)}{\omega_{i}(\P)} \cdot  \round{ \sum\limits_{S \in \S} \frac{\sigma_{i} (S)}{\sigma_{i}(\S)}  \cdot \Inc(S,P)}\\
        &=
        \sum\limits_{P \in \P} \sum\limits_{S \in \S} \frac{\omega_{i} (P)}{\omega_{i}(\P)} \cdot    \frac{\sigma_{i} (S)}{\sigma_{i}(\S)}  \cdot \Inc(S,P)
        =
        \sum\limits_{S \in \S} \frac{\sigma_{i} (S)}{\sigma_{i}(\S)}  \cdot \round{  \sum\limits_{P \in \P}    \frac{\omega_{i} (P)}{\omega_{i}(\P)}  \cdot \Inc(S,P) } \\
       & =
        \EE_{S_i\sim \sigma_i}\bracket{ \sum\limits_{P \in \P} \frac{\omega_{i} (P)}{\omega_{i}(\P)}  \cdot \Inc(S_i,P) ~\bigg|~ F_{i}= \phi}.
    \end{align*}
Since we have equality for each  $\phi \in \F_i$, the total expectations are also equal.
\qed

\section{Proofs of Corollaries}

\subsection{Proof of Corollary \ref{cor:spanpath}}\label{sec:proof_of_corollary_ref_cor_spanpath}

We will create a spanning path as a union of matchings, iteratively constructed on exponentially decreasing input point-sets.

Let $M_1$ be the matching given by \Cref{thm:mainVCmatchings} for the input $(X,\S)$. Let $X_2$ be a set of $\floor{n/2}$ vertices obtained by keeping only one endpoint of each edge in $M_1$. Now we apply \Cref{thm:mainVCmatchings} for the input $(X_2,\S|_{X_2})$ and denote the resulting matching by $M_2$. Observe that $M_1 \cup M_2$ consists of $\floor{n/4}$ disjoint paths of length $3$ (plus at most one smaller component). We continue with $X_3$ which is a set of $\ceil{n/4}$ vertices obtained by keeping only one extremity of each path in $M_1\cup M_2$. Applying \Cref{thm:mainVCmatchings} to $(X_3,\S|_{X_3})$, we obtain a matching $M_3$ such that $M_1 \cup M_2\cup M_3$ consists of $\floor{n/8}$ disjoint paths of length $7$. 
We continue like that for $t=\log n$ iterations, after which $M_1 \cup \dots\cup M_t$ is a single spanning path. 

The crossing number of $P = M_1 \cup \dots\cup M_t$ can be bounded by the sum of crossing numbers of the $M_i$'s. Note that $M_i$ is the result of the $i^{\rm th}$ iterative application of \Cref{thm:mainVCmatchings} where input is $(X_i,\S|_{X_i})$ with $|X_i| = \floor{n/2^{i-1}}$. Therefore, by \Cref{thm:mainVCmatchings}, the crossing number of $M_i$ w.r.t. $\S|_{X_i}$ (and equivalently w.r.t. $\S$) is $O\round{\round{n/2^{i-1}}^{1-1/d}\ln(mn)}$.
Thus, the crossing number of $P$ can be bounded as
\[
    \sum\limits_{i=0}^{\log n} O\round{\round{\frac{n}{2^i}}^{1-1/d}\ln(mn)}
    =
    O\round{n^{1-1/d}\ln(mn)\sum\limits_{i=0}^{\log n}\round{\frac{1}{2^{1-1/d}}}^i}
    =
    \aO\round{n^{1-1/d}}.
\]
Similarly, the total running time can be bounded as
\[
    \sum\limits_{i=0}^{\log n} \aO\round{m\cdot \frac{n}{2^i} +\round{\frac{n}{2^i}}^2 }
    =
    \aO\round{(mn + n^2)\sum\limits_{i=0}^{\log n}\frac{1}{2^i}}
    =
    \aO\round{mn + n^2}.
\]

\subsection{Proof of Corollary \ref{cor:discrepancy}}
There is a direct link between low-crossing matchings and low-discrepancy colorings.

\begin{restatable}{lemma}{matchtodiscLemma}\label{lemma:matching-to-disc}
    Let $(X,\S)$ be a set system, $n = |X|$, $m = |\S|\geq 34$, and let $M$ be a perfect matching of $X$ with crossing number $\kappa$ with respect to $\S$ and for each edge $\{x,y\} \in M$ define
     \[
     \chi_M(x) = \begin{cases} 1 &\text{ with probability } 1/2 \\ -1 &\text{ with probability } 1/2 \end{cases}
     \] 
     and $\chi_M(y) = -\chi_M(x)$.
    Then the expected discrepancy of $\chi_M$ is at most $  \sqrt{3\kappa \ln m} $.
\end{restatable}
\begin{remark*}
A `high probability version' of \Cref{lemma:matching-to-disc} is well-known \cite[Lemma 2.5]{MWW93} and implies the above bound through basic probabilistic calculations. For completeness, we include the proof of \Cref{lemma:matching-to-disc} at the end of this section.
\end{remark*}

\Cref{lemma:matching-to-disc} naturally yields an algorithm to compute a low-discrepancy coloring of any set system with bounded dual-shatter dimension $d$, see \nameref{discalgo} (\Cref{algo:main-discalgo-compressed}).

\smallskip

\begin{algorithm}[t]
\algotitle{\textsc{LowDiscColor}}{discalgo}
\caption{  \textsc{LowDiscColor}$\big((X,\S)\big)$}
    \label{algo:main-discalgo-compressed}

    $n \leftarrow |X|$

    $ \curly{e_1, \dots, e_{\floor{n/2}}} \leftarrow \nameref{GDalgo}\big((X,\S), \floor{ n/2}, 1/2 \big)$ 

    \For{ $i = 1, \dots, \floor{n/2}$}{ 

    $\{x_i,y_i\} \leftarrow \mathrm{endpoints}\round{e_i}$

    $\chi(x_i) = \begin{cases} 1 &\text{ with probability } 1/2 \\ -1 &\text{ with probability } 1/2\end{cases}$

    $\chi(y_i) = - \chi (x_i)$ \tcp*{we skip this step if $y_i = x_i$}

    }
    
    \textbf{return} $\chi$
    
\end{algorithm}

Importantly, constructing the coloring can be done very easily while we construct a matching, therefore the running time of \nameref{discalgo} and  \nameref{GDalgo} are precisely the same. Therefore, \Cref{cor:discrepancy} is an immediate consequence of \Cref{thm:mainVCmatchings} and \Cref{lemma:matching-to-disc}.

\begin{proof}[Proof of \Cref{lemma:matching-to-disc}]

Let $ S \in \S$ be a fixed range. We express the sum $\chi_M(S)$ of colors over elements of $S$ as 
    \[
        \chi_M(S) 
        = \sum\limits_{\{x,y\} \in M; x,y \in S} \round{\chi_M(x) + \chi_M(y)} + \sum\limits_{x \in \mathrm{cr}(S,M)} \chi_M(x) = \sum\limits_{x \in \mathrm{cr}(S,M)} \chi_M(x)~,
    \]
    where $\mathrm{cr}(S,M) = \{x\in S  : \{x,y\} \in M, y \notin S\}$.
    Since $\abs{\mathrm{cr}(S,M)}\leq \kappa$ for any $S \in \S$, $\disc(S,\chi_M)$ is a sum of at most $\kappa$ \emph{independent} random variables. 
    We use the following concentration bound from \cite{alonspencerprobabilistic}
    \begin{claim}[Theorem A.1.1 from \cite{alonspencerprobabilistic}]\label{claim:Alon-concentration}
        Let $X_1, \dots, X_k$ be independent $\{-1,1\}$-valued random variables with $\PP[X_i = -1] = \PP[X_i = 1] = 1/2$. Then for any $\alpha \geq 0$
            \[
                \PP\bracket{\abs{\sum_{i=1}^k X_i} > \alpha} \leq 2e^{-{\alpha^2}/{2k}}.
            \]
    \end{claim}
    \noindent
    Applying \Cref{claim:Alon-concentration}, we get that for any fixed $S \in \S$ and $\alpha > 0$,
        \[
            \PP\bracket{\abs{\chi_M(S)}  > \alpha}  \leq 2e^{-{\alpha^2}/{2\kappa}}.
        \]
    By the union bound, we get
        \[
            \PP \bracket{ \disc_\S(\chi_M) > \alpha } = \PP \bracket{\max\limits_{S \in \S} \abs{\chi_M(S)}> \alpha } \leq m\cdot 2e^{-{\alpha^2}/{2\kappa}}.
        \]
    Finally, we bound the expected discrepancy by applying Fubini's theorem
        \begin{align*}
            &\EE\bracket{ \disc_\S(\chi_M) }
            \overset{\text{def}}{=}
            \int\limits_{0}^\infty \PP \bracket{ \disc_\S(\chi_M) > \alpha } d\alpha
            \leq
            \int\limits_{0}^\infty \min\left\{2m\cdot e^{-{\alpha^2}/{2\kappa}}, ~ 1 \right\} d\alpha\\
            &=
            \int\limits_{0}^{\sqrt{2\kappa\ln(2m)}} 1 d\alpha
            +
            \int\limits_{\sqrt{2\kappa\ln(2m)}}^\infty 2m\cdot e^{-{\alpha^2}/{2\kappa}} d\alpha\\
            &=
            \sqrt{2\kappa\ln(2m)}
            +
            2m\sqrt{2\kappa}\int\limits_{\sqrt{\ln(2m)}}^\infty  e^{-t^2}  ~dt\\
            &=
            \sqrt{2\kappa\ln(2m)}
            +
            2m\sqrt{2\kappa}\int\limits_{\sqrt{\ln(2m)}}^\infty  \frac{t}{t}\cdot e^{-t^2}  ~dt\\
            &\leq
            \sqrt{2\kappa\ln(2m)}
            +
            2m\sqrt{\frac{2\kappa}{\ln(2m)}}\int\limits_{\sqrt{\ln(2m)}}^\infty  te^{-t^2}  ~dt\\
            &=
            \sqrt{2\kappa\ln(2m)}
            +
            2m\sqrt{\frac{2\kappa}{\ln(2m)}}\bracket{  -\frac{e^{-t^2}}{2}}_{\sqrt{\ln(2m)}}^\infty \\
            &=
            \sqrt{2\kappa\ln(2m)}
            +
            \sqrt{\frac{\kappa}{2\ln(2m)}}
            \leq
            \sqrt{3\kappa\ln m},
        \end{align*}
        if $m \geq 34$. This concludes the proof of \Cref{lemma:matching-to-disc}.
\end{proof}

\subsection{Proof of Corollary \ref{cor:GT-neigh}}

We restate and prove~\Cref{cor:GT-neigh} in a self-contained manner, taken from \cite{DEMMPT23}.

\begin{corollary} Let $G = (V, E)$ be a graph on $n$ vertices, and $\eps > 0$ be a parameter such that
the primal shatter function for the set system $\left(V, \left\{N_G[u] \colon u \in V \right\}\right)$
is $O \left( n^{1+\eps} \right)$. Then there exists a family $\K$ of
subsets of $V$ such that
  \begin{enumerate}
    \item[a)] for each $u \in V$, there exists an $F \in \K$ such that $N_G[u] \subseteq F$, 
    \item[b)] the diameter, w.r.t. distances in $G$, of each set of $\K$ is at most $4$, and
    \item[c)] each $u \in V$ lies in $\tilde{O} \left( n^{\eps} \right)$ sets of $\K$.
  \end{enumerate}
The family $\K$ can be computed in time $\tilde{O} \left( n^2 \right)$.
\end{corollary}
\begin{proof}
We first give the construction and its running time. Then we prove that it satisfies
the required three properties. For adjacency-list input, first build an adjacency matrix in $O(n^2)$ time,
so that closed-neighborhood membership tests take constant time.
\begin{description}
\item[Step 1.]
As the neighborhood set system is self-dual, we can
apply \Cref{cor:spanpath} to get a low-crossing spanning path of $V$, which gives an ordering $v_1, \ldots, v_n$
of $V$ such that for each $u \in V$, the set $N_G[u]$ can be written as a union
of $\aO \left( n^{1-1/(1+\eps)} \right) = \aO \left( n^{\eps} \right)$ intervals of this ordering.
Let $\I_u$ be the set of these intervals for $u$, and let $\I = \bigcup_u \I_u$.

As there are $n$ ranges $N_G[u]$, \Cref{cor:spanpath} takes  time $\tilde{O} \left(n^2\right)$.
We do not construct $\I$ explicitly.

\item[Step 2.]
Construct  a greedy covering $\I' \subseteq \I$ as follows: add to $\I'$ an interval $I_1 \in \I$ containing $v_1$ and   covering
the maximum number of other elements of $V$. Say $I_1$ ends at $v_{i_1}$. Next add to $\I'$ an interval $I_2 \in \I$ that covers $v_{i_1+1}$
and the maximum number of remaining uncovered elements of $V$. And so on from left to right until all elements of $V$ are covered.

This can be done in time $O \left(n^2 \right)$ as follows: make a list $L$ of all vertices of $V$ that contain $v_1$ in their closed neighborhood.
Then for the next vertex $v_2$, throw away the vertices in $L$ that do not contain $v_2$ in their closed neighborhood. Continue as long as $L$ is non-empty.
The vertex that stays the longest in $L$ gives the interval $I_1$. If there are several such vertices, we take an arbitrary one. We can continue similarly from the next uncovered vertex. Processing each
vertex in the ordering takes $O(n)$ time, so in total we can construct $\I'$ in time $O \left( n^2 \right)$. Note that if $v_{i_j}$ is the right endpoint of $I_j$, then $I_{j+2}$
starts after $v_{i_j+1}$; otherwise it would have been a better
choice than $I_{j+1}$. Hence each vertex lies in at most two
intervals of $\I'$.

\item[Step 3.] For each $I \in \I'$, add the following one set to $\K$: $F_I = \bigcup_{v \in I} N_G[v]$.

For each $I \in \I'$, constructing $F_I$ takes time $O(|I| \, n)$. The total time is $O( n \sum_{I \in \I'} |I| ) = O \left(n^2 \right)$, as  each vertex of $V$ lies in at most two intervals of $\I'$.
\end{description}

Property a) follows since we include $N_G[u]$ for each $u \in V$ in at least one set of $\K$,
as the intervals in $\I'$ cover $V$.

Property b) is due to the fact that from any $u \in F_I$,
we can go to a vertex in $I$ {\color{myred}in at most one step},
and then to the common vertex that defined the interval $I$
{\color{myred}in at most one more step.
Thus any two vertices of $F_I$ are at distance at most $4$}.

Property c). Fix any $u \in V$.
Due to the greedy construction of $\I'$, an interval in $\I_u$
can only intersect at most {\color{myred}three} intervals in $\I'$.
{\color{myred}
Indeed, if $J \in \I_u$ met both $I_j$ and $I_\ell$
with $\ell \geq j+3$, it would contain $v_{i_j+1}$
and end after $I_{j+1}$, contradicting the choice of $I_{j+1}$.
Since $u \in F_I$ if and only if $I \cap N_G[u] \neq \emptyset$,
}
$u$ can be contained in at most {\color{myred}three} $F_I$ sets
for each interval of $\I_u$.
As $|\I_u|= {\color{myred}\tilde{O}} \left(n^{\eps} \right)$,
the proof follows.

\end{proof}

\subsection{Proof of Corollary \ref{cor:balls}}
Test-sets are small-sized subfamilies (of half-spaces and balls) such that if a matching has low crossing number with respect to this subfamily, then it is guaranteed to have low crossing number with respect to any member of the family (of half-spaces and balls respectively).
It is well known that for set systems induced by half-spaces, one can compute a linear-sized test set in linear time~\cite{Mat92}. As it was observed recently~\cite{CsM21}, one can obtain test-sets for balls via a simple lifting argument, yielding the following guarantee.

\begin{lemma}\cite[Lemma 17]{CsM21}
\label{lemma:test-set-balls}
    Let $X$ be a set of $n$ points in $\RR^d$. There exists a set $\Q_X$ of at most $(d+2)n^{1+1/d}$ balls such that  if a perfect matching of $X$ has crossing number $\kappa$ with respect to $\Q_X$, then its crossing number with respect to any ball in $\RR^d$ is at most $(d+2)\kappa + 6(d+1)^2n^{1-1/d}$.
\end{lemma}

Applying \Cref{thm:mainVCmatchings} to the set system $(X,\Q_X)$ provided by \Cref{lemma:test-set-balls}, we conclude \Cref{cor:balls}.

\subsection{Proof of Corollary \ref{thm:sublinear-delta}}
We prove a more general statement on $t$-partitions which implies \Cref{thm:sublinear-delta} by substituting $t = \floor{n/2}$.
\begin{restatable}{theorem}{sublineardelta} \label{thm:sublinear-delta-partitions}
Given a set system $(X,\S)$ with dual-shatter function $\pi^*_{\S}(k)=O(k^d)$, a positive integer $t\leq n$, and a parameter $\delta \in (0,1)$, we can find a balanced $t$-partition $\{P_1, \dots, P_t\}$ that has expected
 crossing number 
    \[
        \aO\round{ \frac{n^{2}}{t^{1+1/d-\delta}}}
    \]
    in expected time 
    $
    \aO\round{
        \textcolor{myred}{nt_1+nt_2}+mt_1t_2^{1/d}
        +\frac{mnt_2}{t_1^{1/d}}
    }
    $, where $t_1=t^{1-\delta}$ and $t_2=t/t_1$.
\end{restatable}

\textbf{Proof.} 
We will apply a divide-and-conquer approach to construct a partition. 

First, we construct a balanced $t_1$-partition $X = X_1 \cup \dots \cup X_{t_1}$ using \Cref{thm:main} with sampling probability $t_1/n$. Then the idea is that for each $X_i$, we can apply \Cref{thm:main} again to construct a balanced $t_2$-partition $X_i = Y^i_1 \cup \dots \cup Y^i_{t_2}$ of $(X_i,\S|_{X_i})$, where $\S|_{X_i}$ denotes the set of those ranges in $\S$ that cross $X_i$. This then gives a partition $\mathcal Y = \curly{Y^j_k ~|~ (j,k) \in [1,t_1]\times [1,t_2]}$ with $t_1\cdot t_2 = t$ parts. By \Cref{thm:main}, any range $S \in \S$ crosses $\aO\round{n/t_1^{1/d}}$ parts among $X_1, \dots, X_{t_1}$
and 
\begin{equation}\label{eq:crossing-bound-perpart-2ndphase}
\text{within a part $X_i$, ~$S$ crosses at most $\aO\round{|X_i|/t_2^{1/d}}$ parts of $\Y$.}
\end{equation}
As Property~\eqref{eq:crossing-bound-perpart-2ndphase} holds for every $S\in\S$, the expected value of the total crossing number is
\begin{equation}\label{eq:2phase-total-crossing}
    \aO\round{
        \frac{n}{t_1^{1/d}} \cdot \frac{n/t_1}{t_2^{1/d}} 
    }
    =
    \aO\round{
        \frac{n^{2}}{t^{1-\delta+1/d}}
    }.
\end{equation}
Algorithmically, the above procedure is too costly, as constructing the sets of ranges $\S|_{X_i}$ for each $i$ takes $O(mn)$ time overall. Instead, observe that we do not need to keep track of `small' ranges in $\S|_{X_i}$ as the crossing number of any range $S \in \S$ with respect to \emph{any} partition of $X_i$ is at most $\rho(S,X_i) := \min\curly{\abs{S\cap{X_i}},|\overline{S}\cap{X_i}|}$. In particular, if $\rho(S,X_i) \leq |X_i|/t_2^{1/d}$, then Property~\eqref{eq:crossing-bound-perpart-2ndphase} is trivially satisfied for $S$. 

Therefore, after we have created a balanced $t_1$-partition, for each $i\in[1,t_1]$, we call \nameref{algo:GD}$\big((X_i,\mathcal L(\S, X_i)), t_2, t_2/|X_i| \big)$, where
\[
    \mathcal L(\S, X_i):=\curly{ S \in \S ~:~ \rho(S,X_i) > |X_i|/t_2^{1/d}}.
\]
Note that $\mathcal L(\S, X_i)$ does not need to store $S\cap X_i$, only (a pointer to) the range $S$. By the above discussion, the partition constructed this way still has the expected crossing number stated in \Cref{eq:2phase-total-crossing}.


By \Cref{thm:main}, it takes $O\round{t_1(m+n) } = O\round{t^{1-\delta}(m+n)}$ time to construct the first partition. After this, we can recover every range in $\mathcal L(\S, X_i)$ using an $\eps$-net for the set system $\round{X_i, \S\cup\bar{\S}}$ with parameter $\eps = t_2^{-1/d}$.
 As $\S$ has dual shatter dimension bounded by a function of $d$, the VC-dimension of $\S\cup\bar{\S}$ is also bounded by a function of $d$, and so a uniform random sample $N_i$ of size
\[
\Theta \round{t_2^{1/d}\cdot \log \round{t_2^{1/d}}}
\]
is an $\eps$-net with high probability (the constants in $\Theta \left(\cdot\right)$ depend on $d$).

Now observe that if $N_i$ is a $t_2^{-1/d}$-net of $\round{X_i, \S\cup\bar{\S}}$, then $S \in \mathcal L(\S,X_i)$ only if it crosses $N_i$. Therefore, we can call partition on the set  $\widehat{\mathcal L}_i:=
    \{S\in\S:S\text{ crosses }N_i\}$. The advantage is that it is easier to find the elements of $\widehat{\mathcal L}_i$. Indeed it is sufficient to test for each $S \in \S$ and each $x \in N_i$ whether $x \in S$ which takes $\aO\round{t_2^{1/d}m}=\aO\round{t^{\delta/d}m}$ time for one part $X_i$, and $\aO\round{t_1t^{\delta/d}m}=\aO\round{t^{1-\delta+\delta/d}m}$ time in total.
For each $i \in [1,t_1]$, let $m_i = \abs{\widehat{\mathcal L}_i}$.
By \Cref{thm:main}, it takes 
\[
    O\round{
        \sum\limits_{i=1}^{t_1} t_2(m_i+|X_i|) 
    }
    =
    O\round{nt^{\delta} +
       t^{\delta}\cdot \sum\limits_{i=1}^{t_1} m_i 
    }
\]
time to build all the parts in $\Y$.
Recall that by construction, any range in $\S$ crosses $\aO\round{n/t_1^{1/d}}$ parts of the first partition, which implies 
\[
     t^{\delta} \cdot \sum\limits_{i=1}^{t_1} m_i 
    =
    \aO\round{ t^{\delta} \cdot m \cdot \frac{n}{t_1^{1/d}}} 
    =
    \aO\round{ m \cdot \frac{nt^{\delta}}{t^{\frac{1-\delta}{d}}}}.
\]
Therefore, the expected total  running time can be bounded by
$$
\aO\round{
        \textcolor{myred}{nt_1+nt_2}+mt_1t_2^{1/d}
        +\frac{mnt_2}{t_1^{1/d}}
    }.
$$

\qed

\newpage

\section{Automated market makers for prediction markets}\label{sec:Apx-AMM}
In this section we give precise definitions and prove \Cref{cor:AMM,cor:RQRU}. We follow the exposition and techniques of \cite{HWY25}.

\subsection{Preliminaries}

In a prediction market, traders can buy and sell bets (formally called securities) that pay out based on the outcomes of future events. These markets can be modelled with a set system $(X,\S)$, where $X$ denotes the set of possible outcomes and $\S$ contains the \emph{combinations} of outcomes on which traders can place bets. Given a set system and a market state, an automated market maker's (AMM) task is to determine the prices of (combinations of) securities using a cost function and update the prices and market states after each trade.

For example, \textsc{top $L$ candidates}~\cite{HWY25} allows traders to bet on the outcome of top $L$ candidates among $K \geq L$ candidates. Given a set $H \subset [K]$ with $|H| = L$, a top $L$ security for $H$ is the set of permutations where the top $L$ candidates are from $H$. There are a total of $M:= {K \choose L}$ top $L$ securities and $n := K!$ possible outcomes.

Following \cite{HWY25}, we consider automated market makers that determine security prices using a potential function $C: \RR^{|X|} \to \RR$, called a \emph{cost function}. The market state is specified by a vector $\mathbf w \in \RR^{|X|}$, where $w(x)$ denotes the number of shares of security sold by the AMM for the outcome $x \in X$. 

\begin{definition}
Given a set system $(X,\S)$ and a cost function $C: \RR^{|X|} \to \RR$, an AMM for a prediction market with cost function $C$ on $(X,\S)$ securities takes an initial state $\mathbf w^{(0)}  X \to \RR$, offers securities for all $S \in \S$, and supports a sequence of price, cost, and buy operations taking one of the following forms: for any set $ S \in \S$, shares $s \in \RR$, and state $\mathbf w$,
\begin{description}
    \item[Price] operation returns the current price of security for $S$, that is, $\mathrm{price}_C(S) = \sum\limits_{x \in S} \frac{\partial}{\partial \mathbf w(x)} C(\mathbf w)$;
    \item[Cost] operation returns the current cost of $s$ shares of security for $S$, that is,
    \[
        \mathrm{cost}_C(S,s) = C(\mathbf w + s \mathbf 1_S) - C(\mathbf w),
        \text{ where $\mathbf 1_S$ is the indicator vector of $S$;}
    \]
    \item[Buy] operation $\mathrm{buy}_C (S, s)$ updates the state $\mathbf w \leftarrow \mathbf w + s \mathbf 1_S$.
\end{description}
\end{definition}


In particular, we consider two examples of cost functions:

\begin{description}
\item[The logarithmic market scoring rule (LMSR)] proposed by Hanson~\cite{Hanson03} is a popular cost-function-based market making mechanism. It uses the logarithmic scoring rule to interact with a sequence of traders to trade securities and maintain a probability distribution over an outcome space $X$ of $n$ points in an online manner.
\[
    C_L(\mathbf w) = b \ln \round{\sum\limits_{x \in X} \exp\round{\frac{\mathbf w(x)}{b}}},
\]
where $b > 0$ is the liquidity parameter. The market designer can choose $b$ to control how fast the price
moves in response to trading and limit the worst-case loss of the market maker to $b \ln |X |$.
\item[The quadratic scoring rule (QMSR)] is another popular choice of proper scoring rule with the following cost function~\cite{CP07}
\[
    C_Q(\mathbf w) = \frac 1 n \sum\limits_{x \in X} \mathbf w(x) + \frac 1 {4b} \sum\limits_{x \in X} \mathbf w^2(x) - \frac 1 {4bn} \round{\sum\limits_{x \in X} \mathbf w(x)}^2 - \frac b n,
\]
where $b > 0$ is the liquidity parameter.

\end{description}

Hossain~\etal{}~\cite{HWY25} showed that designing efficient AMMs for the above two cost functions is equivalent to designing an efficient data-structure for the following two classical problems.

\begin{definition}[Multiplicative range query and range update]
 Given a set system $(X,\S)$ and a weight function $W : X \to \RR_+$, the  multiplicative range query and range update (multiplicative RQRU) problem asks for an algorithm that performs for any $ S \in \S$ and $ w \in \RR_+$ the following operations: 
\begin{itemize}
    \item $\mathrm{query}(S)$ :  compute and return $\sum\limits_{x \in S} W(x)$;
    \item $\mathrm{update}(S,w)$ : update $W$ with 
     $W(x) =  \begin{cases} w \cdot W(x) \text{ if  } x \in S \\ W(x) \text{ otherwise.}\end{cases}  $
 \end{itemize}
\end{definition}

\begin{definition}[Additive range query and range update]
 Given a set system $(X,\S)$, $l \in \NN$, and a weight function $W : X \to \RR^l$, the additive range query and range update (additive RQRU) problem asks for an algorithm that performs for any $ S \in \S$ and $ w \in \RR^l$ the following operations: 
\begin{itemize}
    \item $\mathrm{query}(S)$ : compute and  return $\sum\limits_{x \in S} W(x)  \in \RR^l$;
    \item $\mathrm{update}(S,w)$ : update $W$ with 
     $W(x) =  \begin{cases} w + W(x) \text{ if  } x \in S \\ W(x) \text{ otherwise.}\end{cases}  $
 \end{itemize}
\end{definition}

\subsection{Partition tree scheme}
The key element in the data-structure constructed by Hossain~\etal{}~\cite{HWY25} is partition trees with low visiting numbers.
\begin{definition}[Partition tree and visiting number]

Given a set system $(X,\S)$, a partition tree is a rooted tree $\T = (V,E)$ where each node $v \in V$ is associated with a canonical subset $N(v) \subseteq X$ and a list of its children $C(v) \subset V$ such that 
\begin{itemize}
    \item for any non-leaf node v, its children’s node-sets form a partition of its node-set;
    \item the node-set of the root $r$ is the base set, that is, $N(r) = X$;
    \item node-sets of leaves are singletons.
\end{itemize}
Given a canonical set $A \subseteq X$, we say that the range $S\in \S$ crosses $A$ if there exist $x,y \in A$ such that $|\{x,y\} \cap S | =1$. We also say that $S$ visits node $v$ if either $v$ is the root of $\T$ or $S$ crosses the node-set of $v$'s parent. The visiting number of the partition tree $\T$ is the maximum number of nodes visited by any single range in $\S$.
\end{definition}
The next theorem summarizes some of the main results of Hossain \etal{}~\cite{HWY25}.
\begin{theorem}[Theorems 3.2 and 4.1 in~\cite{HWY25}]\label{thm:Hossain-main}
Given a set system $(X,\S)$ and a partition tree $\T$ with visiting number $\nu$, there are algorithms that support
\begin{itemize}
\item the query and update operations of the multiplicative (or additive) RQRU problem on $(X,\S)$ with time complexity $O(\nu)$;
\item  price, cost, and buy operations for LMSR (or QMSR) with time complexity $O(\nu)$.
\end{itemize}
\end{theorem}

In fact, one can use low-crossing spanning paths to construct balanced partition trees of low visiting number.

\begin{lemma}[Chazelle and Welzl~\cite{CW89}]
Given a set system $(X,\S)$ and a spanning path $P$ of $X$ with crossing number $\kappa$, one can construct a balanced binary partition tree with visiting number at most $2\kappa \ceil{\log n} +1$ in time $O\round{n \log n}$.
\end{lemma}

Thus, \Cref{cor:spanpath} implies the following statement.

\begin{lemma}\label{lemma:parttree}
Given a set system $(X,\S)$ with dual-shatter function $\pi^*_{\S}(k)=O(k^d)$, one can construct a balanced binary partition tree with visiting number $\aO\round{n^{1-1/d}}$ in time $\aO\round{mn + n^2}$.
\end{lemma}

\Cref{lemma:parttree} and \Cref{thm:Hossain-main} prove \Cref{cor:RQRU,cor:AMM}.

\end{document}